\documentclass[10pt,a4paper]{article}

\usepackage{bookmark}
\usepackage{comment}
\usepackage[utf8]{inputenc}
\usepackage[english]{babel}
\usepackage[T1]{fontenc}
\usepackage{hyperref}
\usepackage{amsmath,amsthm,enumerate}
\usepackage{amssymb}
\usepackage{tikz}
\usepackage{subfigure}
\usepackage[hmargin=2.5cm,vmargin=3cm]{geometry}
\usepackage[color=green!30]{todonotes}

\usepackage{graphicx}
\providecommand{\keywords}[1]
{
  \small	
  \textbf{Keywords: } #1
}

\newtheorem{theorem}{Theorem}

\newtheorem{proposition}[theorem]{Proposition}

\newtheorem{corollary}[theorem]{Corollary}

\newtheorem{definition}[theorem]{Definition}

\newtheorem{claim}{Claim}[theorem]

\newcommand{\smallqed}{{\tiny ($\Box$)}}
\newenvironment{claimproof}{\noindent\emph{Proof of claim.}}{~\smallqed\newline\medskip}

\newcommand{\fullversion}[1]{}

\DeclareMathOperator{\sep}{sep}
\DeclareMathOperator{\sepRB}{sep_{RB}}
\DeclareMathOperator{\maxsepRB}{max-sep_{RB}}

\newcommand{\PBSEP}{\textsc{Red-Blue Separation}}
\newcommand{\PBMAXSEP}{\textsc{Max Red-Blue Separation}}

\newcommand{\x}{\overline{x}}
\newcommand{\y}{\overline{y}}

\newcommand{\Pb}[4]{%
\smallskip
  \noindent\begin{tabular}{|l|}%
  \hline
    \begin{minipage}[c]{\textwidth}
      \smallskip%
      \par\noindent%
      #1%
      \par\noindent%
      \textbf{\textsf{Input}}: #2%
      \par\noindent%
      \textbf{\textsf{#3}}: #4 
      \smallskip%
      \par\noindent%
    \end{minipage}
  \\\hline
  \end{tabular}%
  \smallskip
}%

\title{The RED-BLUE SEPARATION problem on graphs\thanks{This study has been carried out in the frame of the ``Investments for the future'' Programme IdEx Bordeaux - SysNum (ANR-10-IDEX-03-02). Ralf Klasing's research was partially supported by the ANR project TEMPOGRAL (ANR-22-CE48-0001). Florent Foucaud was partially financed by the IFCAM project ``Applications of graph homomorphisms'' (MA/IFCAM/18/39), the ANR project GRALMECO (ANR-21-CE48-0004) and the French government IDEX-ISITE initiative 16-IDEX-0001 (CAP 20-25). Tuomo Lehtil\"a's research was supported by the Finnish Cultural Foundation and by the Academy of Finland grant 338797.}}

\author{Subhadeep Ranjan Dev\footnote{ACM Unit, Indian Statistical Institute, Kolkata, India}
    \and Sanjana Dey \footnote{National University of Singapore, Singapore}~~\footnote{Most of the work of this author was done when she was in Indian Statistical Institute, Kolkata}
    \and Florent Foucaud\footnote{Université Clermont-Auvergne, CNRS, Mines de Saint-Étienne, Clermont-Auvergne-INP, LIMOS, 63000 Clermont-Ferrand, France}~\footnote{Univ. Orléans, INSA Centre Val de Loire, LIFO EA 4022, F-45067 Orléans Cedex 2, France}
    \and Ralf Klasing\footnote{Universit\'e de Bordeaux, Bordeaux INP, CNRS, LaBRI, UMR 5800, Talence, France}
    \and Tuomo Lehtil\"a\footnote{University of Turku, Department of Mathematics and Statistics, Turku, Finland}~~\footnote{Most of the work of this author was done when he was in Univ Lyon, Universit\'e Claude Bernard, CNRS, LIRIS - UMR 5205, F69622, France}
}

\begin{document}
\maketitle

\begin{abstract}
We introduce the \PBSEP{} problem on graphs, where we are given a graph $G = (V, E)$ whose vertices are colored either red or blue, and we want to select a (small) subset $S \subseteq V$, called \emph{red-blue separating set}, such that for every red-blue pair of vertices, there is a vertex $s \in S$ whose closed neighborhood contains exactly one of the two vertices of the pair. 
We study the computational complexity of \PBSEP, in which one asks whether a given red-blue colored graph has a red-blue separating set of size at most a given integer. We prove that the problem is NP-complete even for restricted graph classes. 
We also show that it is always approximable in polynomial time within a factor of $2\ln n$, where $n$ is the input graph's order.
In contrast, for triangle-free graphs and for graphs of bounded maximum degree, we show that \PBSEP{} is solvable in polynomial time when the size of the smaller color class is bounded by a constant. 
However, on general graphs, we show that the problem is $W[2]$-hard even when parameterized by the solution size plus the size of the smaller color class. 
We also consider the problem \PBMAXSEP{} where the coloring is not part of the input. Here, given an input graph $G$, we want to determine the smallest integer~$k$ such that, \emph{for every possible red-blue coloring} of $G$, there is a red-blue separating set of size at most $k$. We derive tight bounds on the cardinality of an optimal solution of \PBMAXSEP, showing that it can range from logarithmic in the graph order, up to the order minus one. We also give bounds with respect to related parameters. For trees however we prove an upper bound of two-thirds the order. We then show that \PBMAXSEP{} is NP-hard, even for graphs of bounded maximum degree, but can be approximated in polynomial time within a factor of $O(\ln^ 2 n)$.
\end{abstract}

\setcounter{footnote}{0}

\section{Introduction}

We introduce and study the \PBSEP{} problem for graphs. 
Separation problems for discrete structures have been studied extensively from various perspectives. In the 1960s, R\'enyi~\cite{R61} introduced the \textsc{Separation} problem for set systems (a set system is a collection of sets over a set of vertices), which has been rediscovered by various authors in different contexts, see e.g.~\cite{B72,CCCHL08,HY14,MS85}. 
In this problem, one aims at selecting a solution subset $\mathcal S$ of sets from the input set system to separate every pair of vertices, in the sense that the subset of $\mathcal S$ corresponding to those sets to which each vertex belongs to, is unique. The graph version of this problem (where the sets of the input set system are the closed neighborhoods of a graph), called \textsc{Identifying Code}~\cite{KCL98}, is also extensively studied. These problems have numerous applications in areas such as monitoring and fault-detection in networks~\cite{UTS04}, biological testing~\cite{MS85}, and machine learning~\cite{KE07}. The \PBSEP{} problem which we study here is a red-blue colored version of \textsc{Separation}, where instead of all pairs we only need to separate red vertices from blue vertices.

In the general version of the \PBSEP{} problem, one is given a set system $(V,\mathcal S)$ consisting of a set $\mathcal S$ of subsets of a set $V$ of vertices which are either blue or red; one wishes to separate every blue from every red vertex using a solution subset $\mathcal C$ of $\mathcal S$ (here a set of $\mathcal C$ separates two vertices if it contains exactly one of them). Motivated by machine learning applications, a geometric-based special case of \PBSEP{} has been studied in the literature, where the vertices of $V$ are points in the plane and the sets of $\mathcal S$ are half-planes~\cite{CN98}. 
The classic problem \textsc{Set Cover} over set systems generalizes both \textsc{Geometric Set Cover} problems and graph problem \textsc{Dominating Set} (similarly, the set system problem \textsc{Separation} generalizes both \textsc{Geometric Discriminating Code} and the graph problem \textsc{Identifying Code}). It thus seems natural to study the graph version of \PBSEP.

\paragraph{Problem definition.} In the graph setting, we are given a graph $G$ and a red-blue coloring $c:V(G)\to \{\text{red},\text{blue}\}$ of its vertices, and we want to select a (small) subset $S$ of vertices, called \emph{red-blue separating set}, such that for every red-blue pair $r,b$ of vertices, there is a vertex from $S$ whose closed neighborhood contains exactly one of $r$ and $b$. Equivalently, $N[r]\cap S\neq N[b]\cap S$, where $N[x]$ denotes the closed neighborhood of vertex $x$; the set $N[x]\cap S$ is called the \emph{code} of $x$ (with respect to $S$), and thus all codes of blue vertices are different from all codes of red vertices. The smallest size of a red-blue separating set of $(G,c)$ is denoted by $\sepRB(G,c)$. Note that if a red and a blue vertex have the same closed neighborhood, they cannot be separated. Thus, for simplicity, we will consider only \emph{twin-free} graphs, that is, graphs where no two vertices have the same closed neighborhood. Also, for a twin-free graph, the vertex set $V(G)$ is always a red-blue separating set as all the vertices have a unique subset of neighbors. We have the following associated computational problem.

\Pb{\PBSEP}
{A red-blue colored twin-free graph $(G,c)$ and an integer $k$.}
{Question}
{Do we have $\sepRB(G,c)\leq k$?}


It is also interesting to study the problem when the red-blue coloring is not part of the input. For a given graph $G$, we thus define the parameter $\maxsepRB(G)$ which denotes the largest size, over each possible red-blue coloring $c$ of $G$, of a smallest red-blue separating set of $(G,c)$. The associated decision problem is stated as follows.

\Pb{\PBMAXSEP}
{A twin-free graph $G$ and an integer $k$.}
{Question}
{Do we have $\maxsepRB(G) \leq k$?} 
In Figure~\ref{fig:sepMaxsep}, to note the difference between $\sepRB$ and $\maxsepRB$, a path of 6 vertices $P_6$ is shown, where the vertices are colored red or blue.
\begin{figure}[h]
    \centering
    \includegraphics[scale = 0.9]{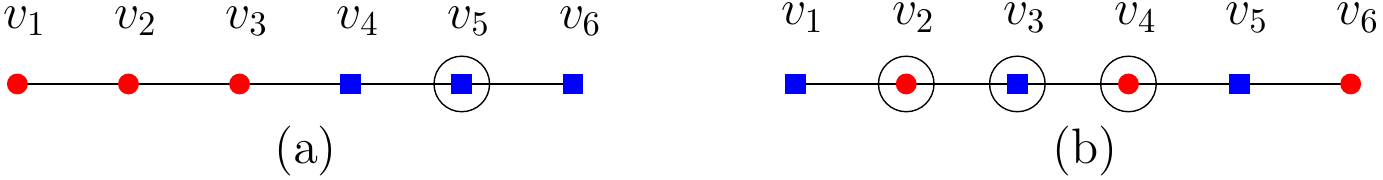}
    \caption{A path of 6 vertices where (a) $\sepRB(P_6, c) = 1$ and (b) $\maxsepRB(P_6) = 3$. The square vertices are blue, the round ones are red; the members of the red-blue separating set are circled.}
    \label{fig:sepMaxsep}
  \end{figure}

\paragraph{Our results.}
We show that \PBSEP{} is NP-complete even for restricted graph classes such as planar bipartite sub-cubic graphs, in the setting where the two color classes\footnote{One class consists of vertices colored \emph{red} and the other class consists of vertices colored \emph{blue}.} have equal size. We also show that the problem is NP-hard to approximate within a factor of $(1-\epsilon)\ln n$ for every $\epsilon>0$, even for split graphs\footnote{A graph $G = (V, E)$ is called a \emph{split graph} when the vertices in $V$ can be be partitioned into an independent set and a clique.} of order $n$, and when one color class has size~$1$. On the other hand, we show that \PBSEP{} is always approximable in polynomial time within a factor of $2\ln n$.
In contrast, for triangle-free graphs and for graphs of bounded maximum degree, we prove that \PBSEP{} is solvable in polynomial time when the smaller color class is bounded by a constant (using algorithms that are in the parameterized class XP, with the size of the smaller color class as parameter). However, on general graphs, the problem is shown to be $W[2]$-hard even when parameterized by the solution size plus the size of the smaller color class. (This is in contrast with the geometric version of separating points by half-planes, for which both parameterizations are known to be fixed-parameter tractable~\cite{BGL19,KMMPS21}.)

As the coloring is not specified, $\maxsepRB(G)$ is a parameter that is worth studying from a structural viewpoint. In particular, we study the possible values for $\maxsepRB(G)$. We show the existence of tight bounds on $\maxsepRB(G)$ in terms of the order $n$ of the graph $G$, proving that it can range from $\lfloor \log_2 n\rfloor$ up to $n-1$ (both bounds are tight). For trees however we prove bounds involving the number of support vertices (i.e. which have a leaf neighbor), which imply that $\maxsepRB(G)\leq\frac{2n}{3}$. We also give bounds in terms of the (non-colored) separation number. We then show that the associated decision problem \PBMAXSEP{} is NP-hard, even for graphs of bounded maximum degree, but can be approximated in polynomial time within a factor of $O(\ln^ 2 n)$.

An extended abstract of this paper was presented in the conference IWOCA'22 and appeared in the proceedings as~\cite{iwoca22}. The present paper includes all the proofs that were missing from the conference version.

\paragraph{Related work.} \PBSEP{} has been studied in the geometric setting of red and blue points in the Euclidean plane~\cite{BGL19,CDKW05,MMS20}. In this problem, one wishes to select a small set of (axis-parallel) lines such that any two red and blue points lie on the two sides of one of the solution lines. The motivation stems from the \textsc{Discretization} problem for two classes and two features in machine learning, where each point represents a data point whose coordinates correspond to the values of the two features, and each color is a data class. The problem is useful in a preprocessing step to transform the continuous features into discrete ones, with the aim of classifying the data points~\cite{CN98,KMMPS21,KE07}. This problem was shown to be NP-hard~\cite{CN98} but 2-approximable~\cite{CDKW05} and fixed-parameter tractable when parameterized by the size of a smallest color class~\cite{BGL19} and by the solution size~\cite{KMMPS21}. A polynomial time algorithm for a special case was recently given in~\cite{MMS20}.

The \textsc{Separation} problem for set systems (also known as \textsc{Test Cover} and \textsc{Discriminating Code}) was introduced in the 1960s~\cite{R61} and widely studied from a combinatorial point of view~\cite{BS07,B72,CCCHL08,HY14} as well as from the algorithmic perspective for the settings of classical, approximation and parameterized algorithms~\cite{CGJMY16,DHHHLRS03,MS85}. The associated graph problem is called \textsc{Identifying Code}~\cite{KCL98} and is also extensively studied (see~\cite{biblio} for an online bibliography with almost 500 references as of January 2022); geometric versions of \textsc{Separation} have been studied as well~\cite{GeomDC,GP19,HJ20}. The \textsc{Separation} problem is also closely related to the \textsc{VC Dimension} problem~\cite{VCdim} which is very important in the context of machine learning. In \textsc{VC Dimension}, for a given set system $(V,\mathcal S)$, one is looking for a (large) set $X$ of vertices that is \emph{shattered}, that is, for every possible subset of $X$, there is a set of $\mathcal S$ whose trace on $X$ is the subset. This can be seen as ''perfectly separating'' a subset of $\mathcal S$ using $X$; see~\cite{BLLPT15} for more details on this connection.

\paragraph{Structure of the paper.} We start with the algorithmic results on \PBSEP{} in Section~\ref{sec:complexity-SEP}. We then present the bounds on $\maxsepRB$ in Section~\ref{sec:bounds} and the hardness result for \PBMAXSEP{} in Section~\ref{sec:complexity-MAXSEP}. We conclude in Section~\ref{sec:conclu}. 


\section{Complexity and algorithms for \PBSEP}\label{sec:complexity-SEP}

We will prove some algorithmic results for \PBSEP{} by reducing to or from the following problems.

\Pb{\textsc{Set Cover}}{A set of elements $U$, a family $\mathcal S$ of subsets of $U$ and an integer $k$.}{Question}{Does their exist a cover $\mathcal C \subseteq \mathcal S$, with $|\mathcal C| \leq k$ such that $\bigcup_{C \in \mathcal C} C = U$?}

\Pb{\textsc{Dominating Set}}{A graph $G = (V, E)$ and an integer $k$.}{Question}{Does there exist a set $D \subseteq V$ of size $k$ with $\forall v \in V, N[v] \cap D \neq \emptyset$?}

\subsection{Hardness}

\begin{theorem}\label{thm:split-hard}
    {\PBSEP} cannot be approximated within a factor of $(1-\epsilon)\cdot\ln n$ for any $\epsilon>0$ even when the smallest color class has size $1$ and the input is a split graph of order $n$, unless P = NP. Moreover, {\PBSEP} is W[2]-hard when parameterized by the solution size together with the size of the smallest color class, even on split graphs.
\end{theorem}
           
\begin{proof}
    For an instance $((U, \mathcal{S}),k)$ of \textsc{Set-Cover}, we construct in polynomial time an instance $((G,c), k)$ of {\PBSEP} where $G$ is a split graph and one color class has size~1. The statement will follow from the hardness of approximating \textsc{Min Set Cover} proved in~\cite{DS13}, and from the fact that \textsc{Set Cover} is W[2]-hard when parameterized by the solution size~\cite{DF99}. 

    We create the graph $(G,c)$ by first creating vertices corresponding to all the sets and the elements. We connect a vertex $u_i$ corresponding to an element $i \in U$ to a vertex $v_j$ corresponding to a set $S_j \in \mathcal{S}$ if $u_i \in S_j$. We color all these vertices blue. We add two isolated blue vertices $b$ and $b'$. 
    We connect all the vertices of type $u_i \in U$ to each other. Also, we add a red vertex $r$ and connect all vertices $u_i \in U$ to $r$. Now, note that the vertices $U \cup \{r\}$ form a clique whereas the vertices $v_j$ along with $b$ and $b'$ form an independent set. Thus, our constructed graph $(G,c)$ with the coloring $c$ is a split graph. See Figure~\ref{fig:setCoverToRBS}.

    \begin{figure}[!htpb]
        \centering
        \includegraphics[scale = 0.8]{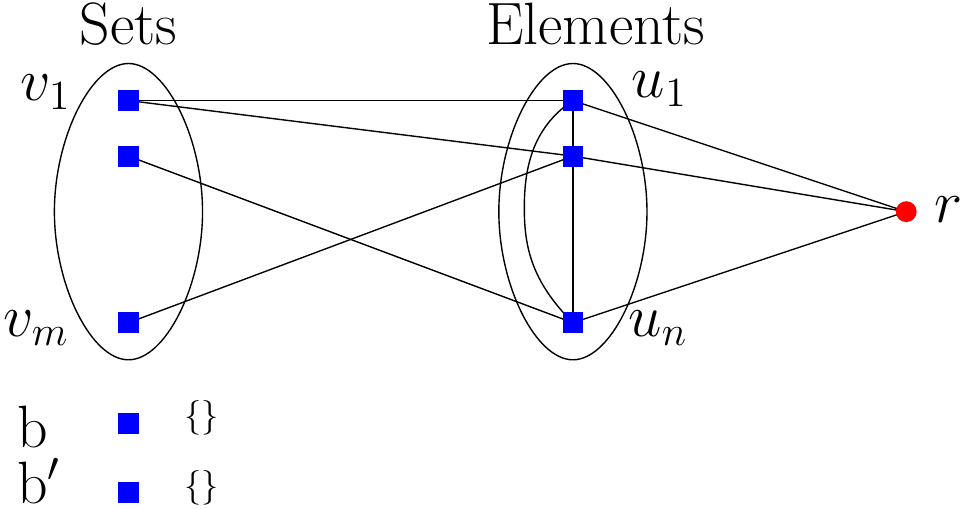}
        \caption{Reduction from \textsc{Set Cover} to {\PBSEP} of Theorem~\ref{thm:split-hard}. Vertex $r$ is red, the others are blue.}
        \label{fig:setCoverToRBS}
    \end{figure}

    \begin{claim}\label{claim:setcover-reduction}
        $\mathcal{S}$ has a set cover of size $k$ if and only if $G$ has a red-blue separating set of size at most $k + 1$.
    \end{claim}
    \begin{claimproof}
      Let $\mathcal{C}$ be a set cover of $(U, \mathcal{S})$ of size $k$. We construct a red-blue separating set $S$ of $(G, c)$ of size at most $k + 1$ as follows. For each set $s_j \in \mathcal{S}$ selected in the set cover $\mathcal{C}$, we choose the corresponding vertex $v_j$ in $S$. Also include the vertex $r$ in $S$. Observe that some blue vertices may have the empty code and the red vertex has itself as the code. Also, the vertices $u_i$ are dominated by the vertices $v_j$ and have some unique code different from $r$. Therefore, $S$ is a separating set of $(G, c)$ of size at most $k + 1$.
        
      Conversely, consider a red-blue separating set $S$ of $(G, c)$ of size $k'$. Since the vertices $U \cup \{r\}$ form a clique, choosing any vertex from this set will not separate any two vertices of the clique. Let us assume that the red vertex $r$ of $(G, c)$ gets the empty code. But then, we have the two isolated blue vertices, one of which also gets the empty code. Thus, the red vertex has to be selected. But $r$, being part of the clique, has to be separated from the blue vertices. Thus, we have to choose vertices from the independent set to separate the blue vertices of the clique from $r$. So, we dominate the blue vertices of the clique by using the vertices in the independent set, which gives our set cover of size at most $k'-1 = k$.
\end{claimproof}
This completes the proof of the theorem.
\end{proof}

\begin{theorem}\label{thm:red-DS}
  \PBSEP{} is NP-hard for bipartite planar sub-cubic graphs of girth at least 12 when the color classes have almost the same size.
\end{theorem}

\begin{proof}
  We reduce from \textsc{Dominating Set}, which is NP-hard for bipartite planar sub-cubic graphs with girth at least 12 that contain some degree-2 vertices~\cite{Zverovich}. We reduce any instance $(G, k)$ of \textsc{Dominating Set} to an instance $((H, c), k')$ of \PBSEP, where $k' = k + 1$ and the number of red and blue vertices in $c$ differ by at most 2.

\begin{figure}[!htpb]
    \centering
    \includegraphics[width = .8 \textwidth]{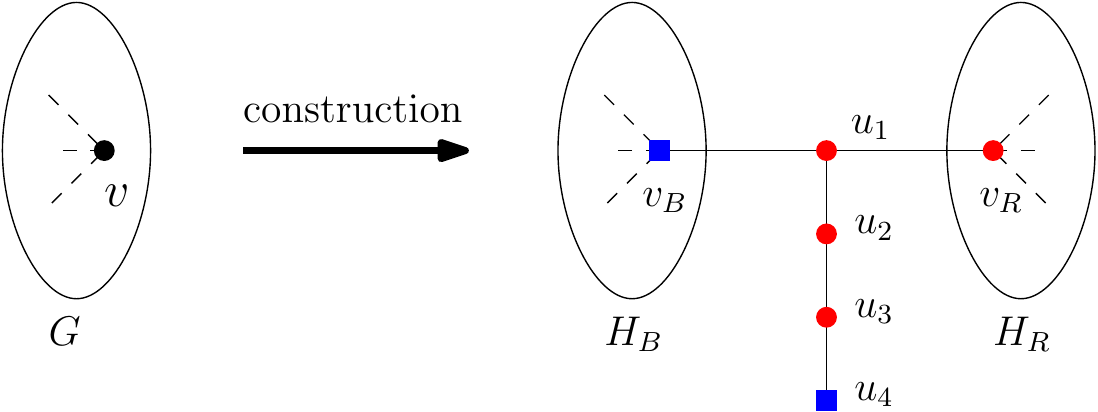}
    \caption{Reduction from \textsc{Dominating Set} to \PBSEP{} of Theorem~\ref{thm:red-DS}. Vertices $v_B$ and $u_4$ are blue, the others are red.}
    \label{fig:rbsreductionequal}
\end{figure}

\paragraph{Construction.} We create two disjoint copies of $G$ namely $H_B$ and $H_R$ and color all vertices of $H_B$ blue and all vertices of $H_R$ red. Select an arbitrary vertex $v$ of degree 2 in $G$  (we may assume such a vertex exists in $G$ by the reduction of \cite{Zverovich}) and look at its corresponding vertices $v_R \in V(H_R)$ and $v_B \in V(H_B)$. We connect $v_R$ and $v_B$ with the head of the path $u_1, u_2, u_3, u_4$ as shown in Figure~\ref{fig:rbsreductionequal}. The tail of the path, i.e. the vertex $u_4$, is colored blue and the remaining three vertices $u_1, u_2$ and $u_3$ are colored red. Our final graph $H$ is the union of $H_R, H_B$ and the path $u_1, u_2, u_3, u_4$ and the coloring $c$ as described. Note that if $G$ is a connected bipartite planar sub-cubic graph of girth at least $g$, then so is $H$ (since $v$ was selected as a vertex of degree 2). We make the following claim. 

\begin{claim}\label{clm1}
The instance $(G, k)$ is a YES-instance of \textsc{Dominating Set} if and only if $\sepRB(H, c) \leq k' = k + 1$.
\end{claim}
\begin{claimproof}
$(\implies)$ Let $D$ be a dominating set of $G$ of size $k$. We construct a red-blue separating set $S$ of $(H, c)$ of size at most $k + 1$ as follows. For each vertex in $D$, include its corresponding vertex in $H_R$ in $S$. Also include in $S$, the vertex $u_2 \in V(H)$. Observe that all blue vertices have the empty code and all red vertices have some non-empty code. Therefore $S$ is a separating set of $(H, c)$ of size at most $k + 1$.

$(\impliedby)$ Consider a red-blue separating set $S$ of $(H, c)$ of size $k'$. Let us assume that the red vertices of $(H, c)$ did not get the empty code. The argument when the blue vertices do not get the empty code is similar. Then, the set $S$ also dominates $V(H_R) \cup \{u_1, u_2, u_3\}$. Since $u_3$ can only be dominated by a vertex $x$ in $\{u_2, u_3, u_4\}$, the vertices in $V(H_R)$ are dominated by $S' = S \setminus \{x\}$. If $S' = S' \cap V(H_G)$ then the set $D$ formed by choosing the corresponding vertices of $S'$ in $G$ is a dominating set of size $k = k' - 1$. Otherwise, $S' \setminus \{u_1\} = S' \cap V(H_R)$, and the set $D$ formed by choosing the corresponding vertices of $S' \setminus \{u_1\}$ in $G$ is a dominating set of $G$ of size $k = k' -1$.\end{claimproof}
This completes the proof of the theorem.
\end{proof}

In the previous reduction, we could choose any class of instances for which \textsc{Dominating Set} is known to be NP-hard. We could also simply take two copies of the original graph and obtain a coloring with two equal color class sizes (but then we obtain a disconnected instance). In contrast, in the geometric setting, the problem is fixed-parameter-tractable when parameterized by the size of the smallest color class~\cite{BGL19}, and by the solution size~\cite{KMMPS21}. It is also 2-approximable~\cite{CDKW05}.


We can also provide another reduction, as follows.

\begin{theorem}
    \PBSEP{} is NP-hard even when the input is a subcubic planar bipartite graph of girth at least 12 and the size of the smallest color class is $k + 1$.
\end{theorem}

\begin{proof}
  We reduce from \textsc{Dominating Set}.

\textsc{Dominating Set} is NP-hard when the input graph is a subcubic planar bipartite graph of girth at least 12~\cite{Zverovich}. We reduce an instance $(G, k)$ of the \textsc{Dominating Set} to an instance $((H, c_{k + 1}), k)$ of \PBSEP{} where $c_{k + 1}$ is a coloring of $H$ with the minimum color class being of size at most $k + 1$.

\begin{figure}[h]
  \centering
  \includegraphics[width = .65 \textwidth]{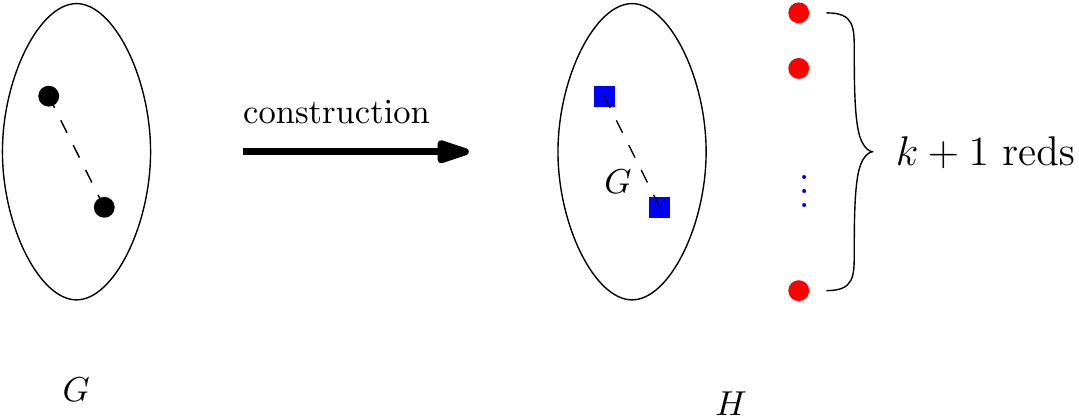}
  \caption{Construction from an instance of \textsc{Dominating Set} to an instance of \PBSEP{} where size of the smaller color class is bounded. The square vertices are blue, the round ones are red.}
  \label{fig:rbsreduction1}
\end{figure}

\paragraph{Construction.} Without loss of generality let us assume the smaller sized color class to be red. For $H$ we create a copy $G_H$ of $G$ with all its vertices colored blue and an independent set $I$ of size $k + 1$ all its vertices colored red vertices. We now make the following claim.


\begin{claim}
  $(G, k)$ is a YES-instance of \textsc{Dominating Set} if and only if $((H, c_{k + 1}), k)$ is a YES-instance of \PBSEP{}.
\end{claim}

\begin{claimproof}
($\implies$) Let $D$ be a dominating set of $G$. We construct a red-blue separating set of $(H, c_{k + 1})$ by selecting the corresponding vertices of $D$ in $G_H$. Since $D$ is a dominating set of $G$, all blue vertices receive a non-empty code. The red vertices on the other hand receive the empty code and we have a valid separating set.

  ($\impliedby$) Let $(H, c_{k + 1})$ have a separating set $C$ of size at most $k$. Since there are $k + 1$ independent red vertices, there will be at least one red vertex which receives the empty code. So, in order for $C$ to be a valid separating set, all blue vertices must receive some non-empty code. This implies that $C \cap V(G_H)$ is a dominating set of $G_H$ of size at most $k$ and which implies a dominating set of $G$ of size at most $k$.
  \end{claimproof}
  
  This completes the proof.
\end{proof}

\subsection{Positive algorithmic results}

We start with a reduction to \textsc{Set Cover} implying an approximation algorithm.
\begin{figure}[h]
  \centering
  \includegraphics[width = .65 \textwidth]{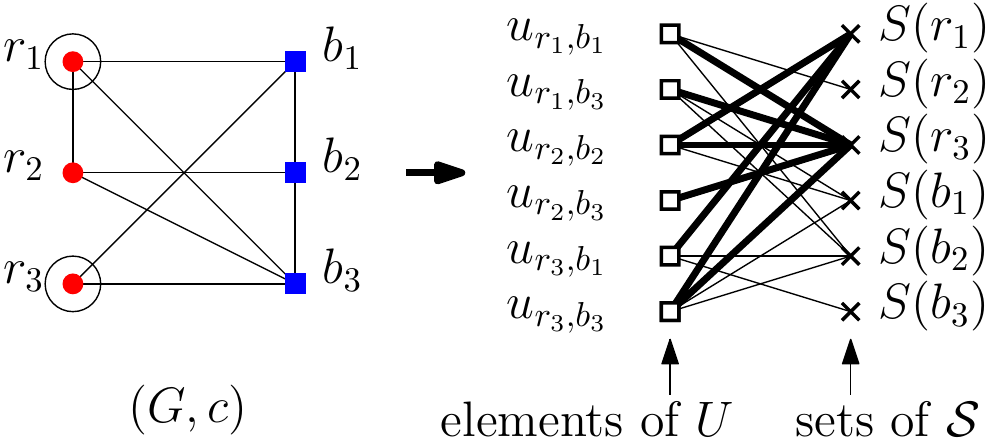}
  \caption{A reduction instance from \PBSEP{} of $(G, c)$ (where the square vertices are blue, the round ones are red) to \textsc{Set Cover} of $(U, \mathcal S)$. The separating set in the colored graph $(G, c)$ and the corresponding set cover in the set system $(U, {\cal S})$ has been highlighted.}
  \label{fig:lognAppx}
\end{figure}

\begin{proposition}\label{rbsep-approx}
  \PBSEP{} has a polynomial-time $(2 \ln n)$-factor approximation algorithm, where $n$ is the input graph's order.
\end{proposition}
\begin{proof}
  We reduce \PBSEP{} to \textsc{Set Cover}. Let $((G, c), k)$ be an input instance of \PBSEP. We reduce it to an instance $((U, \mathcal S), k)$ of \textsc{Set Cover}. For each red-blue vertex pair $(r, b)$ in $G$, create an element in $U$. For each vertex $v$ in $G$ create a set in $\mathcal S$ with elements $(r, b)$ in $U$ such that $v$ is in the closed neighborhood of exactly one of $r$ and $b$ in $G$. See Figure \ref{fig:lognAppx}. Observe that a set cover $\mathcal C$ of size $k$ corresponds to a separating set $S$ of size at most $k$ and vice versa. The greedy algorithm for \textsc{Set Cover} has an approximation factor of $\ln |U| + 1$. Since, in our case $|U| \leq n^2 / 4$, the resulting approximation factor for \PBSEP{} is at most $\ln (n^2 / 4) + 1 \leq 2 \ln n$.
\end{proof}


\begin{proposition}\label{prop:triangle-free}
Let $(G,c)$ be a red-blue colored triangle-free and twin-free graph with the two color classes $R$ and $B$. Then, $\sepRB(G,c)\leq 3\min\{|R|,|B|\}$.
\end{proposition}
\begin{proof}
 Without loss of generality, we assume $|R| \leq |B|$. We construct a red-blue separating set $S$ of $(G,c)$. First, we add all red vertices to $S$. It remains to separate every red vertex from its blue neighbors. If a red vertex $v$ has at least two neighbors, we add (any) two such neighbors to $S$. Since $G$ is triangle-free, no blue neighbor of $v$ is in the closed neighborhood of both these neighbors of $v$, and thus $v$ is separated from all its neighbors (see vertex $v_1$ in Figure \ref{fig:triangleFree}). If $v$ had only one neighbor $w$, and it was blue, then we separate $w$ from $v$ by adding one arbitrary neighbor of $w$ (other than $v$) to $S$. Since $G$ is triangle-free, $v$ and $w$ are separated (see vertex $v_2$ in Figure \ref{fig:triangleFree}). Thus, we have built a red-blue separating set $S$ of size at most $3|R|$.\end{proof}

\begin{figure}[h]
  \centering
  \includegraphics[width = .3 \textwidth]{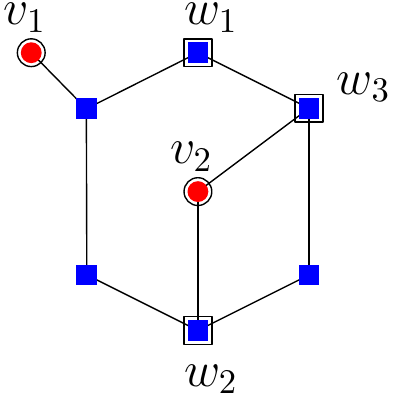}
  \caption{Illustration of Proposition \ref{prop:triangle-free}: here $v_1$ has just one blue neighbor hence $w_1$ is added in $S$. $v_2$'s neighbors $w_2$ and $w_3$ are also included in $S$. The square vertices are blue, the round ones are red.}
  \label{fig:triangleFree}
\end{figure}

\begin{figure}
  \centering
  \includegraphics[width = .5 \textwidth]{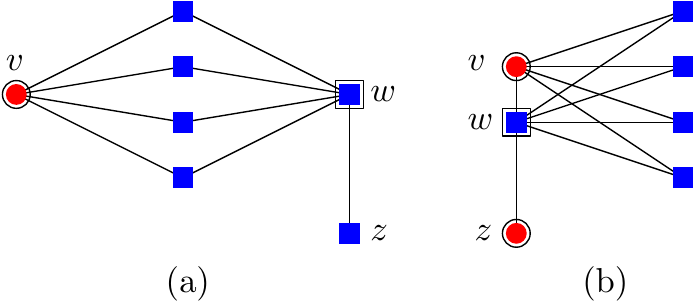}
  \caption{Illustration of Proposition \ref{prop:bounded-degree}: the two cases when the vertices $v$ and $w$ are (a) not adjacent and (b) adjacent. The square vertices are blue, the round ones are red.}
  \label{fig:boundedAlgo}
\end{figure}

\begin{proposition}\label{prop:bounded-degree}
Let $(G,c)$ be a red-blue colored twin-free graph with maximum degree~$\Delta\geq 3$. Then, $\sepRB(G,c)\leq\Delta\min\{|R|,|B|\}$.
\end{proposition}
\begin{proof}
Without loss of generality, let us assume $|R| \leq |B|$. We construct a red-blue separating set $S$ of $(G,c)$. Let $v$ be any red vertex. If there is a blue vertex $w$ whose closed neighborhood contains all neighbors of $v$ ($w$ could be a neighbor of $v$), we add both $v$ and $w$ to $S$ (see Figure \ref{fig:boundedAlgo}(a)). If $v$ is adjacent to $w$, since they cannot be twins, there must be a vertex $z$ that can separate $v$ and $w$; we add $z$ to $S$ (see Figure \ref{fig:boundedAlgo}(b)). Now, $v$ is separated from every blue vertex in $G$.

If such a vertex $w$ does not exist, then we add all neighbors of $v$ to $S$. Now again, $v$ is separated from every vertex of $G$. Thus, we have built a red-blue separating set $S$ of size at most $\Delta|R|$.
\end{proof}

The previous propositions imply that \PBSEP{} can be solved in XP time for the parameter ''size of a smallest color class'' on triangle-free graphs and on graphs of bounded degree (by a brute-force search algorithm). This is in contrast with the fact that in general graphs, it remains hard even when the smallest color class has size~1 by Theorem~\ref{thm:split-hard}.

\begin{theorem} 
  \PBSEP{} on graphs whose vertices belong to the color classes $R$ and $B$ can be solved in time $O(n^{3\min\{|R|,|B|\}})$ on triangle-free graphs and in time $O(n^{\Delta\min\{|R|,|B|\}})$ on graphs of maximum degree $\Delta$.
\end{theorem}


\section{Extremal values and bounds for $\maxsepRB$}\label{sec:bounds}

We denote by $\sep(G)$ the smallest size of a (non-colored) separating set of $G$, that is, a set that separates \emph{all} pairs of vertices. We will use the relation $\maxsepRB(G)\leq \sep(G)$, which clearly holds for every twin-free graph $G$.

\subsection{Lower bounds for general graphs}

We can have a large twin-free colored graph with solution size~2 (for example, in a large blue path with a single red vertex, two vertices suffice). We show that in every twin-free graph, there is always a coloring that requires a large solution.

\begin{theorem}\label{thm:LB}
For any twin-free graph $G$ of order $n\geq 1$ and $n\not\in \{8,9,16,17\}$, we have $\maxsepRB(G)\geq \lfloor \log_2(n)\rfloor$.
\end{theorem}
\begin{proof}
Let $G$ be a twin-free graph of order $n$ with $\maxsepRB(G)=k$. 
There are $2^{n}$ different red-blue colorings of $G$. For each such coloring $c$, we have $\sepRB(G,c)\leq k$. Consider the set of vertex subsets of $G$ which are separating sets of size $k$ for some red-blue colorings of $G$. Notice that each red-blue coloring has a separating set of cardinality $k$. There are at most ${n\choose k}\leq n^k$ such sets.

Consider such a separating set $S$ and consider the set $I(S)$ of subsets $S'$ of $S$ for which there exists a vertex $v$ of $G$ with  $N[v]\cap S=S'$. Let $i_S$ be the number of these subsets: we have $i_S\leq 2^{|S|}\leq 2^k$. If $S$ is a separating set for $(G,c)$, then all vertices having the same intersection between their closed neighborhood and $S$ must receive the same color by $c$. Thus, there are at most $2^{i_S}\leq 2^{2^k}$ red-blue colorings of $G$ for which $S$ is a separating set. 
Hence, we have

$$
\begin{array}{rcl}
2^n & \leq & \binom{n}{k} 2^{2^k}\text{, which implies}\\
2^n & \leq & n^k 2^{2^k}\text{, which implies}\\
n & \leq &k\log_2(n)+2^k
\end{array}
$$

We now claim that this implies that $k\geq \log_2(n-\log_2(n)\log_2(n))$. Suppose to the contrary that this is not the case, then we would obtain:\vspace{-3mm}

$$
\begin{array}{rcl}
n & < & \log_2(n-\log_2(n)\log_2(n))\log_2(n) + n-\log_2(n)\log_2(n)\\
n & < & \log_2(n)\log_2(n) + n-\log_2(n)\log_2(n)
\end{array}
$$

And thus $n<n$, a contradiction. Since $k$ is an integer, we actually have $k\geq\lceil\log_2(n-\log_2(n)\log_2(n))\rceil$. To conclude, one can check that whenever $n\geq 70$, 
we have $\lceil\log_2(n-\log_2(n)\log_2(n))\rceil\geq \lfloor\log_2(n)\rfloor$. Moreover, if we compute values for $2^n-\binom{n}{k} 2^{2^k}$ when $1\leq n\leq 69$ and $k=\lfloor\log_2(n)\rfloor-1$, then we observe that this is negative only when $n\in\{8,9,16,17\}$. Thus, $\lfloor\log_2(n)\rfloor$ is a lower bound for $\maxsepRB(G)$ as long as $n\not\in \{8,9,16,17\}$.
\end{proof}

The bound of Theorem~\ref{thm:LB} is tight for infinitely many values of $n$.

\begin{proposition}\label{pro:2kexample}
For any integers $k\geq 1$ and $n=2^k$, there exists a graph $G$ of order $n$ with $\maxsepRB(G)=k$.
\end{proposition}
\begin{proof}
We build $G$ as follows. Let $S=\{s_1,\ldots,s_k\}$ be a set of $k$ vertices. Let $T$ be another set of vertices (disjoint from $S$). For every subset $S'$ of $S$ of size at least~$2$, we add a vertex $v_{S'}$ to $T$ and we join it to all vertices of $S'$ and $T$. The set $T$ induces a clique. Finally, we add an isolated vertex $v_\emptyset$ to $G$. To see that $\maxsepRB(G)\leq k$, notice that $S$ is a separating set of $G$ (regardless of the coloring).

If $k=1$, the coloring with $s_1$ Red and $v_\emptyset$ Blue shows that $\maxsepRB(G)\geq 1$.

If $k=2$, color $s_1$ and $s_2$ Red and $v_{\{1,2\}}$ and $v_\emptyset$ Blue. To separate $v_{\{1,2\}}$ from $s_1$ (respectively $s_2$), $s_2$ must belong to any separating set (respectively $s_1$), which shows that $\maxsepRB(G)\geq 2$.

If $k\geq 3$, color $v_{S}$ Blue and the other vertices Red. For each subset $S'$ of $S$ with $|S'|=k-1$, in order to separate $v_{S}$ from $v_{S'}$, any separating set needs to contain $s_i$ where $\{s_i\}=S\setminus S'$. This shows that $\maxsepRB(G)\geq k$.
%
\end{proof}

We next relate parameter $\maxsepRB$ to other graph parameters. 

\begin{theorem}\label{ratio}
Let $G$ be a graph on $n$ vertices. Then, $\sep(G)\leq\min\{\lceil\log_2(n)\rceil\cdot\maxsepRB(G),\lceil\log_2(\Delta(G)+1)\rceil\cdot\maxsepRB(G)+\gamma(G)\}$, where $\gamma(G)$ is the domination number of $G$ and $\Delta(G)$ its maximum degree.

\end{theorem}
\begin{proof}
Let $G$ be a graph on $2^{k-1}+1\leq n\leq 2^k$ vertices for some integer $k$. We denote each vertex by a different $k$-length binary word $x_1x_2\cdots x_k$ where each $x_i\in\{0,1\}$. Moreover, we give $k$ different red-blue colorings $c_1,\dots, c_k$ such that vertex $x_1x_2\cdots x_k$ is red in coloring $c_i$ if and only if $x_i=0$ and blue otherwise. For each $i$, let $S_i$ be an optimal red-blue separating set of $(G,c_i)$. We have $|S_i|\leq \maxsepRB(G)$ for each $i$. Let $S=\bigcup_{i=1}^k S_i$. Now, $|S|\leq k\cdot \maxsepRB(G)=\lceil\log_2(n)\rceil \cdot\maxsepRB(G)$. We claim that $S$ is a separating set of $G$. Assume to the contrary that for two vertices $x=x_1x_2\cdots x_k$ and $y=y_1y_2\cdots y_k$, $N[x]\cap S=N[y]\cap S$. For some $i$, we have $y_i\neq x_i$. Thus, in coloring $c_i$, vertices $x$ and $y$ have different colors and hence, there is a vertex $s\in c_i$ such that $s\in N[y]\triangle N[x]$, a contradiction which proves the first bound.

Let $S$ be an optimal red-blue separating set for such a coloring $c$ and let $D$ be a minimum-size dominating set in $G$; $S\cup D$ is also a red-blue separating set for coloring $c$. At most $\Delta(G)+1$ vertices of $G$ may have the same closed neighborhood in $D$. Thus, we may again choose $\lceil\log_2(\Delta(G)+1)\rceil$ colorings and optimal separating sets for these colorings, each coloring (roughly) halving the number of vertices having the same vertices in the intersection of separating set and their closed neighborhoods. Since each of these sets has size at most $\maxsepRB(G)$, we get the second bound.
\end{proof}

We do not know whether the previous bound is reached, but as seen next, there are graphs $G$ such that $\sep(G)=2\maxsepRB(G)$.

\begin{proposition}\label{prop:complete-partite}
Let $G=K_{k_1,\dots,k_t}$ be a complete $t$-partite graph for $t\geq2$, $k_i\geq5$ odd for each $i$. Then $\sep(G)=n-t$ and $\maxsepRB(G)=(n-t)/2$.
\end{proposition}
\begin{proof}
Let us have $t$ parts $G_1,\dots, G_t$ with vertex sets $V(G_i)=\{v^i_1,\dots,v^i_{k_i}\}$. Let $c$ be such a coloring in $G$ that $\sepRB(G,c)=\maxsepRB(G)$. Observe that $S$ is a separating set in $G$ if and only if $|S\cap V(G_i)|\geq k_i-1$. Indeed, if we have $v^i_j,v^i_h\not\in S$, then $N[v^i_j]\cap S=N[v^i_h]\cap S$ since $N(v^i_j)=N(v^i_h)$. Moreover, each vertex in $S$ is separated from vertices not in $S$ and each vertex in $V(G_i)$ is separated from vertices in $V(G_j)$. Thus, $\sep(G)=n-t$.

Next we form a red-blue separating set $S'$ for coloring $c$. For each part, we choose to set $S'$ every vertex which has the less common color in that part (however, we choose at least two vertices to $S'$ from each part). Observe that $|S'|\leq\sum_{i=1}^t\frac{k_i-1}{2}=(n-t)/2$. Moreover as above, we can see  that $S'$ is a red-blue separating set with the help of fact that if $v^i_j\not \in S'$ and $v^i_h\not\in S'$, then $v^i_j$ and $v^i_h$ have the same color.

Finally, we show that $\maxsepRB(G)\geq(n-t)/2$. Observe that if we have $u,v\in V(G_i)$ and $u,v\not\in S$, then $N[u]\cap S=N[v]\cap S$. Thus, any two vertices $u,v\in V(G_i)\setminus S$ must have the same color. If each part has $(k_i+1)/2$ red vertices and $(k_i-1)/2$ blue vertices, then we require at least $\sum_{i=1}^t(k_1-1)/2=(n-t)/2$ vertices in $S$.\end{proof}

\subsection{Upper bound for general graphs}

We will use the following classic theorem in combinatorics to show that we can always spare one vertex in the solution of \PBMAXSEP.

\begin{theorem}[Bondy's Theorem \cite{B72}]
Let $V$ be an $n$-set with a family $\mathcal{A}=\{\mathcal{A}_1,\mathcal{A}_2,\ldots,\mathcal{A}_n\}$ of $n$ distinct subsets of $V$. There is an $(n-1)$-subset $X$ of $V$ such that the sets $\mathcal{A}_1\cap X, \mathcal{A}_2\cap X, \mathcal{A}_3\cap X,\ldots, \mathcal{A}_n\cap X$ are still distinct.
\end{theorem}

\begin{corollary}\label{cor:n-1}
For any twin-free graph  $G$ on $n$ vertices, we have $\maxsepRB(G)\leq \sep(G)\leq n-1$.
\end{corollary}
\begin{proof}
Regardless of the coloring, by Bondy's theorem, we can always find a set of size $n-1$ that separates \emph{all} pairs of vertices. 
\end{proof}

This bound is tight for every even $n$ for complements of \emph{half-graphs} (studied in the context of identifying codes in~\cite{FGKNPV1l}).

\begin{definition}[Half-graph \cite{EH84}]
For any integer $k\geq 1$, the \emph{half-graph} $H_k$ is the bipartite graph on vertex sets $\{v_1,\ldots,v_k\}$ and $\{w_1,\ldots,w_k\}$, with an edge between $v_i$ and $w_j$ if and only if $i\leq j$.

The complement $\overline{H_{k}}$ of $H_k$ thus consists of two cliques $\{v_1,\ldots,v_k\}$ and $\{w_1,\ldots,w_k\}$ and with an edge between $v_i$ and $w_j$ if and only if $i>j$.
\end{definition}

\begin{proposition}\label{prop:half-graph}
For every $k\geq 1$, we have $\maxsepRB(\overline{H_{k}})=2k-1$.
\end{proposition}
\begin{proof}
The upper bound follows from Corollary~\ref{cor:n-1}.

Consider the red-blue coloring $c$ such that $v_i$ is Blue whenever $i$ is odd and Red, otherwise. If $k$ is odd, $w_i$ is Red whenever $i$ is odd and Blue, otherwise. If $k$ is even, $w_i$ is Blue whenever $i$ is odd and Red, otherwise.

For any integer $i$ between $1$ and $k-1$, $v_i$ and $v_{i+1}$ have different colors and can only be separated by $w_{i}$. Likewise, $w_i$ and $w_{i+1}$ have different colors and can only be separated by $v_{i+1}$. This shows that $\{w_1,\ldots, w_{k-1}\}$ and $\{v_2,\ldots, v_{k}\}$ must belong to any separating set of $(\overline{H_{k}}, c)$. Finally, consider $w_1$ and $v_k$. They also have different colors and can only be separated by either $v_1$ or $w_k$. This shows that we need at least $n-1$ vertices in any separating set.
\end{proof}

\subsection{Upper bound for trees}

We will now show that a much better upper bound holds for trees.

Degree-1 vertices are called \textit{leaves} and the set of leaves of the tree $T$ is $L(T)$. Vertices adjacent to leaves are called \textit{support vertices}, and the set of support vertices of $T$ is denoted $S(T)$. We denote $\ell(T)=|L(T)|$ and $s(T)=|S(T)|$. The set of support vertices with exactly $i$ adjacent leaves is denoted $S_i(T)$ and the set of leaves adjacent to support vertices in $S_i(T)$ is denoted $L_i(T)$. Observe that $|L_1(T)|=|S_1(T)|$. Moreover, let $L_+(T)=L(T) \setminus L_1(T)$ and $S_+(T)=S(T)\setminus S_1(T)$. We denote the sizes of these four types of sets by $s_i(T),\ell_i(T), s_+(T)$ and $\ell_+(T)$, respectively.


  For trees, we can show the following result, which is in contrast with the situation for general graphs (or even split graphs, as highlighted by Theorem~\ref{thm:split-hard}).
\begin{theorem}Let $T$ be a tree on $n\geq3$ vertices and let $c$ be a coloring with exactly one red (or blue) vertex. We have $\sepRB(T,c)\leq2$.
\end{theorem}
\begin{proof}
Let $T$ be a tree on at least three vertices with coloring $c$ such that there is exactly one red vertex $v\in V(T)$.

Let us assume first that $v\not\in L(T)$. Thus, $v$ has at least two neighbors $w$ and $u$. If we now include $w$ and $u$ in the separating set $S$, then $v$ is the only vertex in $T$ which has two adjacent vertices in $S$ and hence, $S$ is a red-blue separating set in $T$ for coloring $c$.

On the other hand, if $v\in L(T)$, $u\in N(v)$ and $w\in N(u)\setminus\{v\}$, then $S=\{v,w\}$ is a red-blue separating set in $T$ for coloring $c$, and $\sepRB(T,c)\leq2.$ 
\end{proof}

To prove our upper bound for trees, we need Theorems \ref{thm:TreeBound} and \ref{thm:IDn-sbound}.
\begin{figure}[h]
  \centering
  \includegraphics[scale = .6]{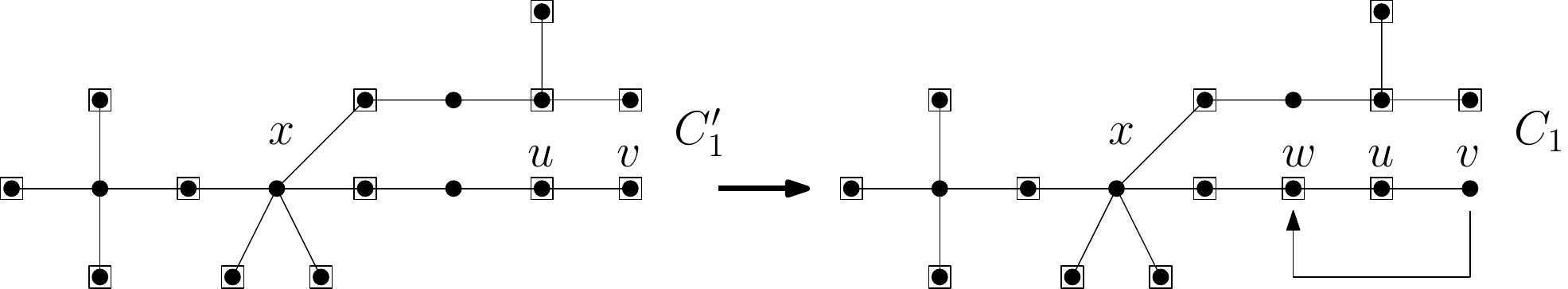}
  \caption{Construction of $C_1$ from $C_1'$ where the boxed elements represent
  members of the set.}
  \label{fig:comparisonc1andc1}
\end{figure}
\begin{theorem}\label{thm:TreeBound}
For any tree $T$ of order $n\geq5$, we have $\maxsepRB(T)\leq\frac{n+s(T)}{2}$.
\end{theorem}
\begin{proof}
Observe that the claim holds for stars (select the vertices of the smallest color class among the leaves, and at least two leaves). Thus, we assume that $s(T)\geq 2$. Let $c$ be a coloring of $T$ such that $\maxsepRB(T)=\sepRB(T,c)$.

We build two separating sets $C_1$ and $C_2$; the idea is that one of them is small. We choose a non-leaf vertex $x$ and add to the first set $C_1'$ every vertex at odd distance from $x$ and every leaf. If there is a support vertex $u\in S_1(T)\cap C_1'$ and an adjacent leaf $v\in L_1(T)\cap N(u)$, we create a separating set $C_1$ from $C_1'$ by shifting the vertex away from leaf $v$ to some vertex $w\in N(u)\setminus L(T)$. We construct in a similar manner sets $C_2'$ and $C_2$, except that we add the vertices at even distance from $x$ to $C_2'$ (including $x$ itself) and do the shifting when $u\in S_1(T)$ has even distance to $x$. Sets $C_1$ and $C_2$ have been previously considered in \cite[Theorem 6]{FT22}. See Figure \ref{fig:comparisonc1andc1}.

\begin{figure}[h]
  \centering
  \includegraphics[width = .7 \textwidth]{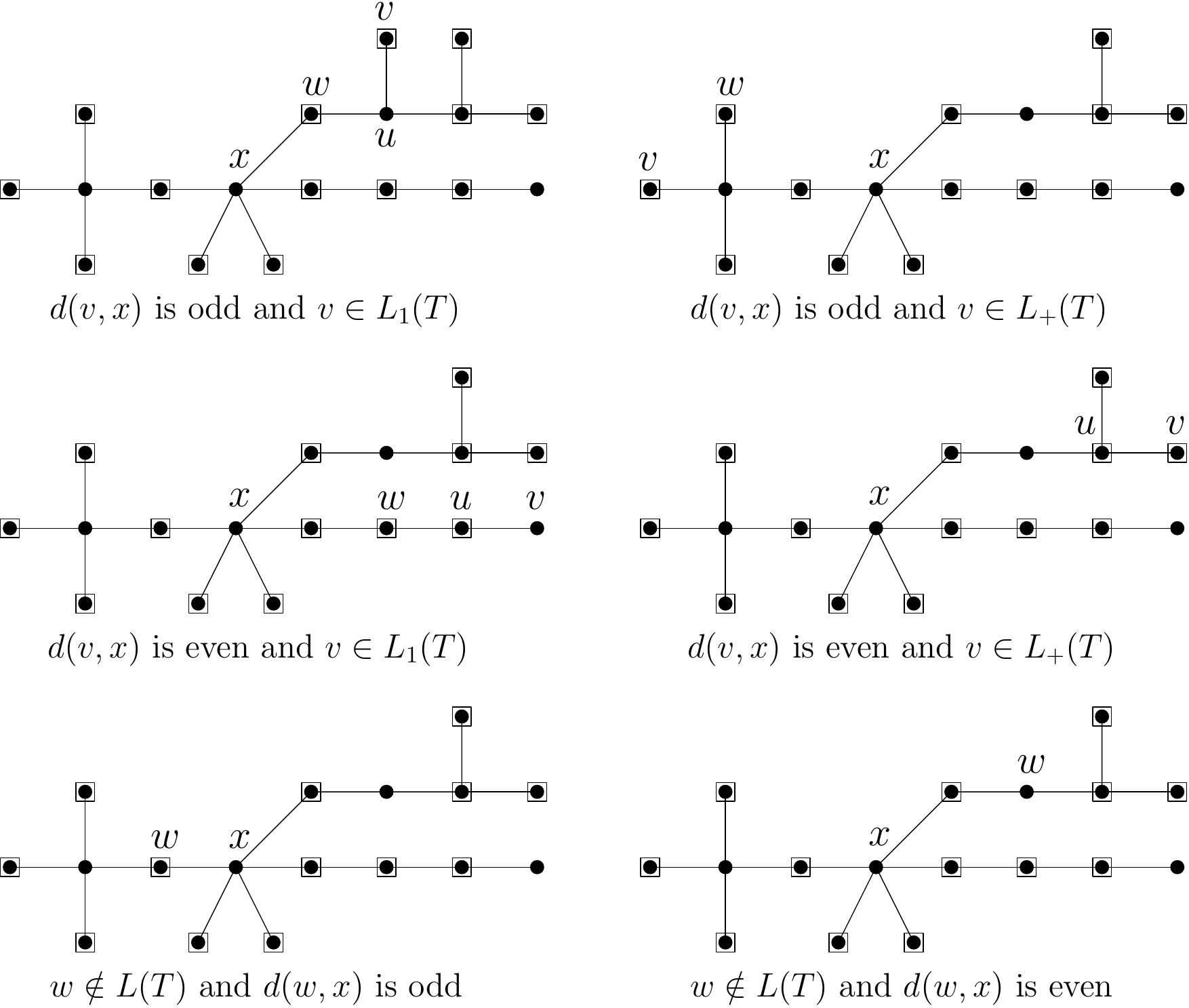}
  \caption{Illustration of various cases for $C_1$ in Claim \ref{Claim:TreeSepSet}.}
  \label{fig:sepSetsCases}
\end{figure}

\begin{claim}\label{Claim:TreeSepSet}
Both $C_1$ and $C_2$ are separating sets.
\end{claim}
\begin{claimproof}
Let us first consider set $C_1$. See Figure \ref{fig:sepSetsCases} for helpful illustrations of different cases we go through in the following arguments. Let us first show that each leaf $v\in L(T)$ is separated from all other vertices by $C_1$. Let $u\in N(v)$ be the adjacent support vertex. If $d(x,v)$ is odd, then $v\in C_1$ and $|N(u)\cap C_1|\geq2$. Hence, $v$ is separated from every other vertex. If $d(v,x)$ is even and $v\in L_+(T)$, then $v\in C_1$ and $|N[u]\cap C_1|\geq3$. Hence, $v$ is again separated from other vertices. If $d(v,x)$ is even and $v\in L_1(T)$, then $v\not\in C_1$, $u\in C_1$ and $|N[u]\cap C_1|\geq2$ due to the shift when we form $C_1$ from $C_1'$. If we have $N[w]\cap C_1=\{u\}$ for some $w\neq u$ and $w\not\in L(T)$, then $d(w,x)$ is even and we have $|N[w]\cap C_1|\geq2$, a contradiction. Thus, any leaf $v$ is separated from each other vertex.

Let us then consider a vertex $w\not\in L(T)$ with odd distance $d(w,x)$, then $w\in C_1$. Moreover, each neighbor $u\in N(w)$ has even distance to $x$ and is either a leaf which is separated from $w$ or has at least two neighbors with odd distances to $x$ and hence, $|C_1\cap N(u)|\geq2$ and $u$ is separated from $w$. Finally, if $w\not\in L(T)$ and $d(w,x)$ is even, then $w\not\in C_1'$ and $|N(w)\cap C_1|\geq2$. Since there are no $3$-cycles nor $4$-cycles in a tree, $w$ is separated from all other vertices.

The proof that also $C_2$ is  a separating set is similar. We only swap evens and odds. Hence, the claim follows. 
\end{claimproof}

Let us denote by $NS_3(T) $ a smallest set of vertices in $T$ such that for each vertex $v\in S_3(T)$ which has $N(v)\cap S_+(T)=\emptyset$, we have at least one vertex $u\in N(v)\setminus L(T)$ in $NS_3(T)$ (such a set exists since $T$ is not a star).

\begin{figure}[h]
  \centering
  \includegraphics[scale = .7]{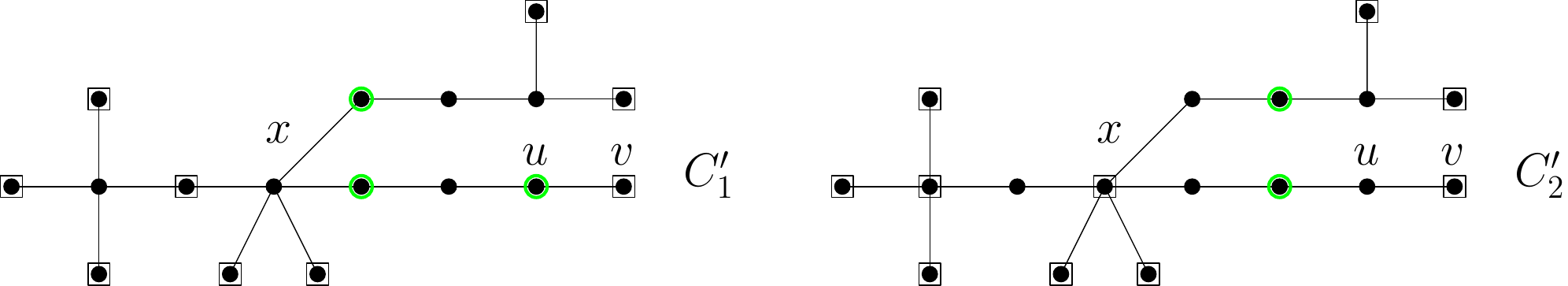}
  \caption{Comparison of the sets $C_1'$ and $C_2'$ where the vertices highlighted in green belong to the sets $C'_1\setminus (L(T)\cup S_+(T)\cup NS_3(T))$ and $C'_2\setminus (L(T)\cup S_+(T)\cup NS_3(T))$, respectively.}
  \label{fig:comparisonc1c2}
\end{figure}

We assume that out of the two sets $C_1'$ and $C_2'$, $C_a'$ ($a\in\{1,2\}$) has less vertices among the vertices in $V(T)\setminus (L(T)\cup S_+(T)\cup NS_3(T))$.   In particular, it contains at most half of those vertices and we have $|C_a'\setminus (L(T)\cup S_+(T)\cup NS_3(T))|\leq (n-\ell(T)-s_+(T)-|NS_3(T)|)/2$. In Figure \ref{fig:comparisonc1c2}, a comparison of the sets $C_1'$ and $C_2'$ is shown. Next, we will construct set $C$ from $C_a'$. Let us start by having each vertex in $C_a'$ be in $C$. Let us then, for each support vertex $u\in S_+(T)$, remove from $C$ every adjacent leaf $w\in L_+(T)\cap N(u)$ such that $w$ is in the more common color class within the vertices in $N(u)\cap L_+(T)$ in coloring $c$. We then add some vertices to $C$ as follows. For $u\in S_i(T)$, $i\geq4$, we add $u$ to $C$ and some leaves so that there are at least two vertices in $N(u)\cap C$. We have at most $|L(T)\cap N[u]|/2+1$ vertices in $C\cap(N[u]\cap L(T)\cup\{u\})$.

For $i=3$, we add $u$ and any $v\in NS_3(T)\cap N(u)\setminus C$, depending on which one already belongs to $C$. Then, if all leaves in $N(u)$ have the same color, we add one of them to $C$. Hence, we have $|C\cap (L_3(T)\cup NS_3(T))|/s_3(T)\leq 2$.

Finally, for $i=2$, if the two leaves have same color and $u\not\in C_a'$, we add $u$ and one of the two leaves to $C$. If the two leaves have the same color and $u\in C_a'$, we add a non-leaf neighbor of $u$ to $C$. If the leaves have different colors, one of them, say $v$, has the same color as $u$. We add $u$ to $C$ and shift the vertex in $C$ in the leaves so that $v$ is in $C$. We added at most two vertices to $C$ in this case. Notice that now we have $S_+(T)\subseteq C$.

Each time, we added to $C$ at most half of the considered vertices in $N(u)$, and at most one other additional vertex. After these changes, we shift some vertices in $C$ away from $L_1(T)$ the same way we built $C_a$ from $C_a'$. As $|C_a'\setminus (L(T)\cup S_+(T)\cup NS_3(T))|\leq (n-\ell(T)-s_+(T)-|NS_3(T)|)/2$, we get:
\begin{flalign*}
|C| & \leq \frac{n-\ell(T)-s_+(T)-|NS_3(T)|}{2}+\ell_1(T)+\frac{\ell_+(T)+|NS_3(T)|}{2}+s_+(T) \\ 
& = \frac{n+\ell_1(T)+s_+(T)}{2}=\frac{n+s(T)}{2}.
\end{flalign*}

\begin{claim}\label{claim:C-is-a-sep-set}
$C$ is a red-blue separating set for coloring $c$.
\end{claim}
\begin{claimproof}
Since $C_a$ is a separating set and $C_a\setminus C\subseteq L_+(T)$, if two vertices $w,v$ are not separated by $C$, then they were separated by a leaf in $L_+(T)$ in $C_a$. 
Moreover, there exists a support vertex $u\in S_+(T)$ such that $v,w\in N[u]$. Recall that $S_+(T)\subseteq C$ and hence, $u\in C$. If $w$ and $v$ are both leaves, then they have the same color and do not need to be separated. In the following, we go through all the other possibilities for $w$ and $v$.

Assume first that $w,v\not\in L(T)$ and $d(w,v)=1$ 
 and let us say, without loss of generality, that $d(w,x)$ is of such parity that $w\in C_a'$ (and $w\in C$). Notice that we have $w=u$ or $v=u$ in this case. Moreover, since there are no cycles in $T$, the parities of $d(w,x)$ and $d(v,x)$ differ. Thus, there exists some vertex $b\in N(v)\setminus N[w]$ such that the parity of $d(b,x)$ equals to the parity of $d(w,x)$. Thus, $b\in C_a'$.  If $b$ is not a leaf, then $b\in C$ and $b$ separates $w$ and $v$. Thus, $b\in L(T)$. However, if $b\in L_1(T)$, then $b\in C$ and if $b\in L_+(T)$, then $v\in S_+(T)$ and there exists a leaf in $C$ adjacent to $v$ separating $v$ and $w$, a contradiction. When $d(w,v)=2$, neither of $w$ or $v$ are in $C$ since they cannot be separated. Again, we can find some vertex $b$ in $C$ that will separate them as in the previous case.

 Moreover, since there are no cycles in $T$, the parities of $d(w,x)$ and $d(v,x)$ differ. Thus, there exists some vertex $b\in N(v)\setminus N[w]$ such that the parity of $d(b,x)$ equals to the parity of $d(w,x)$. If $b$ is not a leaf, then $b\in C$ and $b$ separates $w$ and $v$. Thus, $b\in L(T)$. However, if $b\in L_1(T)$, then $b\in C$ and if $b\in L_+(T)$, then $v\in S_+(T)$ and there exists a leaf in $C$ adjacent to $v$ separating $v$ and $w$, a contradiction. When $d(w,v)\geq2$, we notice that this case cannot occur since $v,w\in N[u]$ and at least one of these two vertices is adjacent to a leaf in $N(u)$.

Finally, we have the case where exactly one of the two vertices is a leaf, let us say $v\in L(T)$ and since $v\in N(u)$, we have $v\in L_+(T)$. Assume first that $d(v,w)=2$. Thus, $v,w\not\in C$. However, then we again have $b\in C$ such that $b\in N(w)\setminus \{u\}$, either due to the parity of $d(x,b)$ or because $b$ is a leaf. As the last case we have $u=w$ and $v\in L_+(T)\cap N(u)$. If $u\not\in C_a'$, then, by parity, $u$ has an adjacent non-leaf vertex in $C$ since $T$ is not a star. On the other hand, if $u\in C_a'$, then if $u\in S_2(T)$ and the two leaves have the same color, there is again a non-leaf neighbor in $C$. If $u\in S_2(T)$ and the two leaves have different colors, then there is a leaf of the same color, in $N(u)$, as $u$ which is in $C$. If $u\in S_3(T)$, then $u$ has a non-leaf neighbor in $NS_3(T)\cap C$ or in $S_+(T)\cap C$. If $u\in S_i(T)$ for $i\geq4$, then $u$ has at least two adjacent leaves which are in $C$. Hence, $w$ and $v$ are either separated or they have the same color.
\end{claimproof}

This completes the proof of the theorem.
\end{proof}

The upper bound of Theorem~\ref{thm:TreeBound} is tight. Consider, for example, a path on eight vertices. Also, the trees presented in Proposition~\ref{prop:3n/5} are within $1/2$ from this upper bound. 
In the following theorem, we offer another upper bound for trees which is useful when the number of support vertices is large. Following theorem has been previously considered for total dominating identifying codes (under term differentiating-total dominating set) in \cite{HHH06}.

\begin{theorem}\label{thm:IDn-sbound}
For any tree $T$ of order $n\geq5$, $\sep(T)\leq n-s(T)$.
\end{theorem}
\begin{proof}
Let us choose for each support vertex $u\in S(T)$ exactly one adjacent leaf $v\in L(T)$ and say that these vertices form the set $S'$. 
Next, we form the separating set $S=V(T)\setminus S'$. Notice that $|S|=n-s(T)$. In the following, we show that $S$ is a separating set in $T$. 

Observe that if $v\not\in S$, then $v$ is a leaf and no support vertex has two adjacent leaves which do not belong to $S$. Thus, vertices which do not belong to $S$ are pairwise separated.
Since $S'\subseteq L(T)$, $T[S]$ is a connected induced subgraph of $T$. Moreover, as $n\geq5$, we have $|N[w]\cap S|\geq2$ for each vertex $w\in S$. 
Thus, vertices in $S$ are separated from vertices which are not in $S$. Finally, any two vertices $w,w'\in S$ are separated since $|V(T[S])|\geq3$; hence, each closed neighborhood is unique in $T[S]$.
\end{proof}

The following corollary is a direct consequence of Theorems \ref{thm:TreeBound} and \ref{thm:IDn-sbound}. Indeed, we have $\maxsepRB(T)\leq \min\{n-s(T),(n+s(T))/2\}$.

\begin{corollary}\label{cor:2n/3}
For any tree $T$ of order $n\geq5$, we have $\maxsepRB(T)\leq\frac{2n}{3}$.
\end{corollary}

We next show that Corollary~\ref{cor:2n/3} (and Theorem~\ref{thm:TreeBound}) is not far from tight.

\begin{proposition}\label{prop:3n/5}
For any $k\geq 1$, there is a tree $T$ of order $n=5k+1$ with $\maxsepRB(T)=\frac{3(n-1)}{5}=\frac{n+s(T)-1}{2}$.
\end{proposition}
\begin{proof}
Consider the tree $T$ formed by taking $k$ disjoint copies $P^1,\ldots, P^k$ of a path of order~$6$ and identifying one endpoint of each path into one single vertex $x$. We consider the coloring that colors $x$ red, and all other vertices are colored red-blue-red... following the bipartition of the tree. Let $v^i_1,\ldots,v^i_5$ be the vertices of $P^i$ distinct from $x$, where $x$ is adjacent to $v^i_1$. In order to separate $v^i_5$ from $v^i_4$, we need $v^i_3$ in any red-blue separating set. To separate $v^i_4$ from $v^i_3$, we need either $v^i_2$ or $v^i_5$. To separate $v^i_3$ from $v^i_2$, we need either $v^i_1$ or $v^i_4$. Thus, we need at least three vertices of $P^i$ in any red-blue separating set, which shows that $\maxsepRB(T)\geq 3k$. To see that $\maxsepRB(T)\leq 3k$, one can see that the set consisting of all vertices $v_1^i,v_3^i,v_5^i$ separates all pairs of vertices. Finally, since $s(T)=k$ and $n-1=5k$, we get that $\frac{n+s(T)-1}{2}=3k=\frac{3(n-1)}{5}$.
\end{proof}

\section{Algorithmic results for \PBMAXSEP} \label{sec:complexity-MAXSEP}

The problem \PBMAXSEP{}  does not seem to be naturally in the class NP (it is in the second level of the polynomial hierarchy). Nevertheless, we show that it is NP-hard. 

\begin{theorem}\label{thm:MAX-RB-NP-h}
  \PBMAXSEP{} is NP-hard even for graphs of maximum degree 12.
\end{theorem}
\begin{proof} We reduce from the following NP-hard version of \textsc{3-SAT}~\cite{T84}.

\Pb{\textsc{3-SAT-2l}}{A set of $m$ clauses $C = \{c_1, \ldots, c_m\}$ each with at most three literals, over $n$ Boolean variables $X = \{x_1, \ldots, x_n\}$, and each literal appears at most twice.}
{Question}{Is there an assignment of $X$ where each clause has a true literal?}

\paragraph{Construction:} Given an instance $\sigma$ of \textsc{3-SAT-2l} with $m$ clauses and $n$ variables, we create an instance $(H, k)$ of \PBMAXSEP, where $k = 4m + 9n$. We first explain the construction of a domination gadget.

A \emph{domination gadget} on vertices $v_1$ and $v_2$, denoted by $H(v_1, v_2)$, consists of 16 vertices including $v_1$ and $v_2$. This gadget will ensure that $v_1,v_2$ are dominated by the local solution of the gadget (and thus, separated from the rest of the graph). See Figure~\ref{fig:dominationa_gadget} for reference. The vertices $v_1$ and $v_2$ may be connected to each other or to some other vertices outside the domination gadget as represented by the dashed edges incident to them. Both $v_1$ and $v_2$ are also connected to the vertices $u_1$ to $u_4$ as shown in the figure. Next we have a clique $K_{10}$ consisting of the vertices $\{p_1, p_2, \ldots, p_6, q_1, q_2, \ldots, q_4\}$. Every vertex $p_i$ is connected to a unique pair of vertices from $\{u_1, u_2, u_3, u_4\}$ and every vertex $q_j$ is connected to a unique triple of vertices from $\{u_1, u_2, u_3, u_4\}$. For example in the figure we have $p_4$ connected with the pair of vertices $u_2$ and $u_3$ and $q_3$ connected with the triplet of vertices $u_1, u_3$ and $u_4$.

\begin{figure}[h!]
  \centering
  \includegraphics[width = .8 \textwidth]{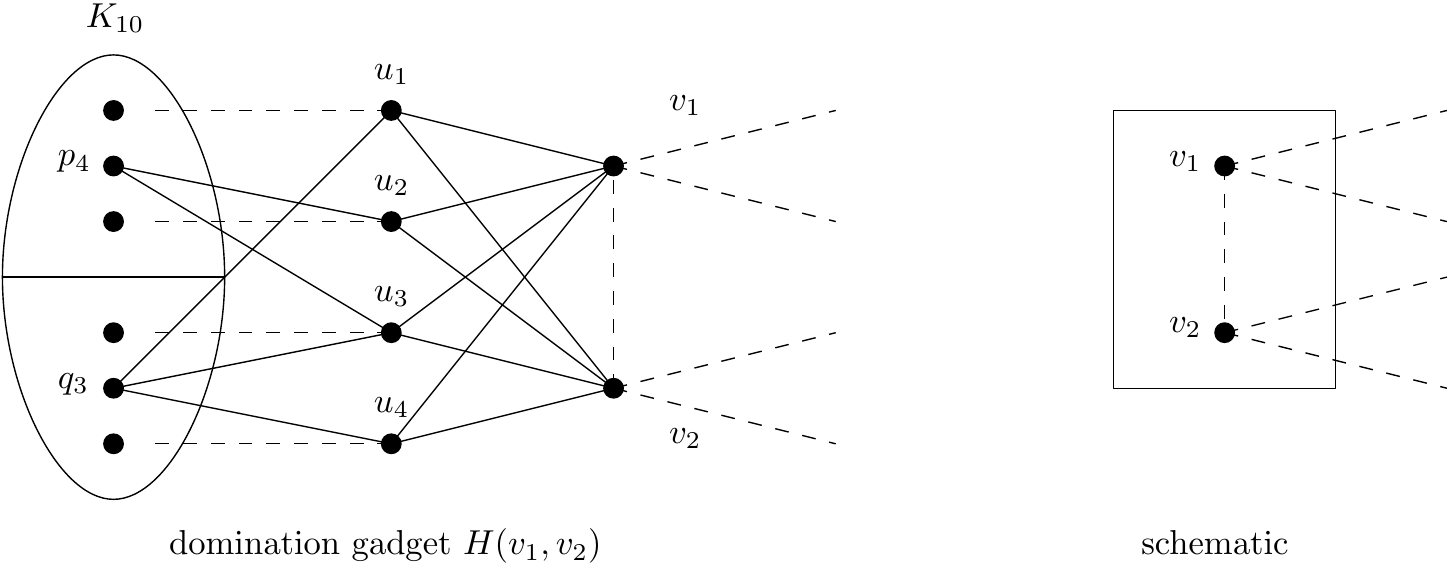}
  \caption{A domination gadget on vertices $v_1$ and $v_2$ and its schematic representation.}
  \label{fig:dominationa_gadget}
\end{figure}

The graph $H$ consists of one \emph{variable gadget} per variable and one \emph{clause gadget} per clause. The variable gadget for a variable $x$ consists of the graph $H(x^a, x^b)$ and $H(x, \x)$ with additional edges $(x^a, x^b), (x^a, x)$ and $(x^a, \x)$. The clause gadget for a clause $C = (x \lor \y \lor z)$ (say) is $H(c^a, c^b)$, where $c^a$ is connected to the vertices $c^b, x, \y$ and $z$. See Figure~\ref{fig:maxhardness} for an illustration. Observe that since each literal can appear at most twice therefore, the maximum degree of $H$ is 12, determined by the $q_i$'s in each domination gadget.

\begin{figure}[h!]
  \centering
  \includegraphics[width = .65 \textwidth]{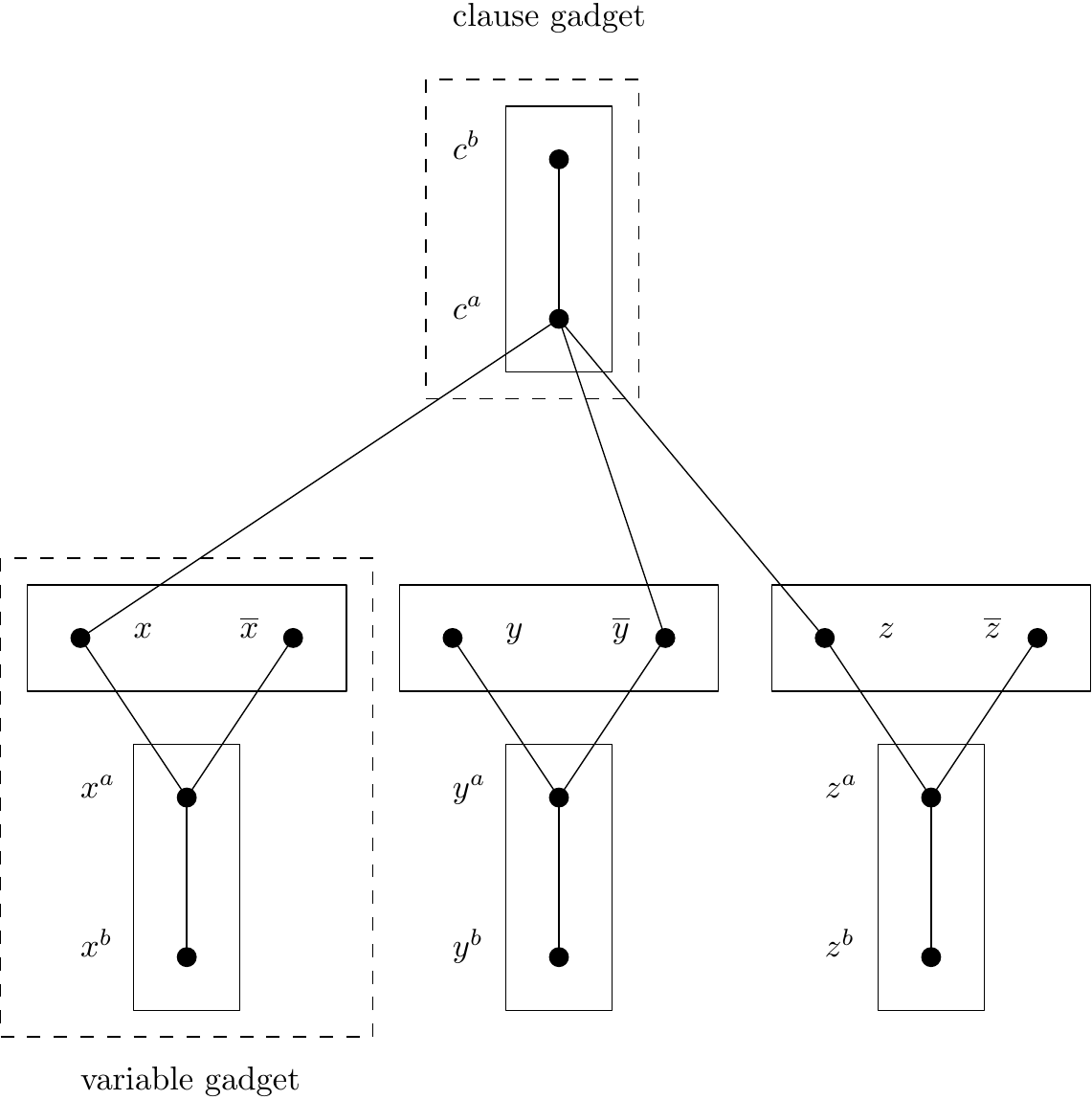}
  \caption{Variable and clause gadgets.}
  \label{fig:maxhardness}
\end{figure}

    Let us illustrate an example instance for the reduction (see Figure~\ref{fig:maxhardnessExample}). In this example, we have the formula as $ (\overline{w} \lor \overline{x} \lor y) \land (x \lor \overline{y} \lor z)$. Thus, we have two clauses $c_1 = \overline{w} \lor \overline{x} \lor y$ and $c_2 = x \lor \overline{y} \lor z$ and four variables $w, x, y$ and $z$. Corresponding to each clause and each variable, we have a clause gadget and a variable gadget respectively (as described earlier). The domination gadget attached with both the clause and variable gadgets is shown below to increase readability. We have shown one possible coloring in this instance with a large solution size.
    
    \begin{figure}[!h]
        \centering
        \includegraphics[width = \textwidth]{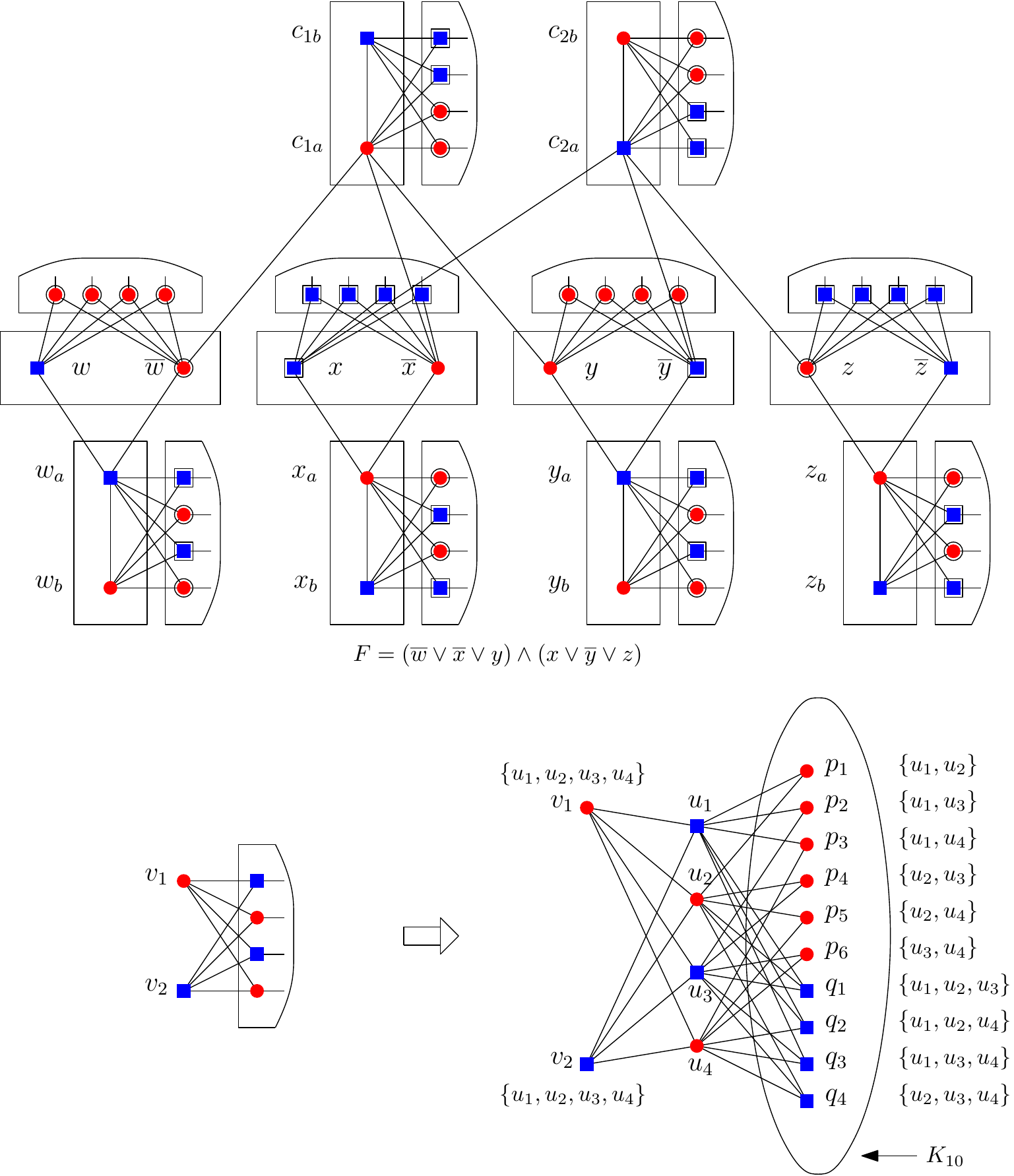}
        \caption{An example to show the reduction from \textsc{3-SAT-2l} to \PBMAXSEP{}. The square vertices are blue, the round ones are red.}
        \label{fig:maxhardnessExample}
    \end{figure}

\begin{claim} \label{claim:maxrbsep-sep}
  $\maxsepRB(H) = \sep(H)$.
\end{claim}

\begin{claimproof}
It is clear that $\maxsepRB(H) \leq \sep(H)$, since any set that separates \emph{all} pairs of vertices also separates all red-blue pairs, for any possible coloring.

We will now prove that $\sep(H)\leq\maxsepRB(H)$ by providing a specific coloring, and showing that for this coloring, any (optimal) solution also separates all pairs of vertices.

  Consider a coloring $c$ of $H$, where for each domination gadget $H(v_1, v_2)$, all $p_i$'s are colored blue and all $q_j$'s are colored red, or vice versa. For each variable $x$, the vertices $x^a$ and $x^b$ are given different colors. For each clause $C$, the vertices $c^a$ and $c^b$ are also colored differently. The rest of the vertices are colored arbitrarily.

  Consider a (minimum) red-blue separating set $S$ of $H$ with respect to the coloring $c$. For each domination gadget $H(v_1, v_2)$, the vertices $u_1$ to $u_4$ must be included in $S$. This is because for each $u_h \in H(v_1, v_2)$, there exist $p_i, q_j \in H(v_1, v_2)$ such that $N[p_i] \triangle N[q_j] = \{u_h\}$, where $\triangle$ denotes the symmetric difference between two sets. Since $p_i$ and $q_j$ are given different colors, $u_h$ is the only vertex which separates them and has to be included in $S$.

  Observe that, for each variable $x$, the vertices $u_1$ to $u_4$ of the domination gadget $H(x^a, x^b)$ separate the vertices $x^a$ and $x^b$ from the rest of $H$. Therefore, they only need to be separated between themselves. The same reasoning holds for the vertices $x$ and $\x$, for each variable $x$, as well as for the vertices $c^a$ and $c^b$, for each clause $C$.

  For each variable $x$, at least one of the vertices $x$ or $\x$ must be included in $S$ in order to separate the vertices $x^a$ and $x^b$ as they are colored differently. This also separates vertices $x$ and $\x$, if they are assigned different colors. For each clause $C = (x \lor \y \lor z)$ (say), at least one of the vertices $x, \y$ or $z$ needs to be included in $S$ (if not already included) in order to separate the vertices $c^a$ and $c^b$ as they are also colored differently. The set $S$ must contain a subset of vertices from the set $\{x, \x \mid x\in X\}$ which separates the vertices $y^a$ and $y^b$ for all variables $y$, and the vertices $c^a$ and $c^b$ for all clauses $C$.

  Since the size of set $S$ is minimum, it is a minimum red-blue separating set of $H$ with respect to the coloring $c$. We will now show that $S$ is also a separating set of (uncolored) $H$. Observe that for each domination gadget $H(v_1, v_2)$, the vertices from $u_1$ to $u_4$ (which are in $S$) separate all (uncolored) $p_i$'s and $q_j$'s from the rest of the vertices in $H$. This is because the intersection of the set $\{u_1, u_2, u_3, u_4\}$ and the closed neighborhood of each $p_i$ or $q_j$ is unique. The vertices $u_1$ to $u_4$ are also separated from the remaining vertices.

  Also, since the vertices $u_1$ to $u_4$ are included in $S$ for all domination gadgets of $H$, therefore going by the previous explanation we only need to separate the vertex pairs $(x^a, x^b), (x, \x)$ and $(c^a, c^b)$. But the construction of $S$ is such that all these pairs of vertices are already separated by $S$. Therefore $S$ is also a separating set of $H$. 
  This proves our claim.
\end{claimproof}

\begin{claim} \label{claim:sep-k}
  If $\sigma$ is satisfiable, $\sep(H) = k$ and otherwise, $\sep(H) > k$, where $k = 4m + 9n$.
\end{claim}

\begin{claimproof}
  Consider a separating set $S$ of $H$. For each domination gadget $H(v_1, v_2)$, the vertices $u_1$ to $u_4$ must be included in $S$ as these are the only vertices which can separate all the $p_i$'s and $q_j$'s of $H(v_1, v_2)$ and themselves. Also, for each variable $x$, one of the vertices $x$ or $\x$ must also be included in $S$ in order to separate $x^a$ and $x^b$. Thus, a total of $4 (m + 2n) + n = 4m + 9n$ vertices are needed in $S$. This implies that $\sep(H) \geq k = 4m + 9n$.
  
  If $\sigma$ is satisfiable, then with respect to a satisfying assignment, we include for each variable $x$, the vertex $x$ if it is assigned true, or the vertex $\x$, if it is assigned false. These vertices also separate $c^a$ and $c^b$ for all clauses $C$. Therefore, $S$ is a separating set of $H$ and $\sep(H) = k$.

  On the contrary, if $\sigma$ is not satisfiable, then $S$ is not a separating set of $H$ as there exists some clause $C$ for which $c^a$ and $c^b$ are not separated. Since all vertices in $S$ are necessary to be included, any separating set of $H$ is a strict superset of $S$, and $\sep(H) > k$.
\end{claimproof}

  From Claim \ref{claim:maxrbsep-sep} and Claim \ref{claim:sep-k} it follows that $\sigma$ is satisfiable if and only if $\maxsepRB(H) \leq k = 4m + 9n$. Since the maximum degree of $H$ is 12 and \textsc{3-SAT-2l} is NP-hard, \PBMAXSEP{} is also NP-hard for graphs with maximum degree 12.
\end{proof}

We can use Theorem~\ref{ratio} and a reduction to \textsc{Set Cover} to show the following algorithmic result.

\begin{theorem}\label{thm:approx-MAX}
  \PBMAXSEP{} can be approximated within a factor of $O(\ln^2 n)$ on graphs of order $n$ in polynomial time.
\end{theorem}

\begin{proof}
  Let $A$ be a polynomial-time ($2\ln n +1$)-approximation algorithm for the \textsc{Separation} problem~\cite{GKM08}. For any graph $G$, let $S(G)$ denote the separating set returned by $A$ on the input graph $G$. Using Theorem~\ref{ratio}, we have
  
  \begin{align*}
    |S(G)| & \leq (2\ln n +1) \cdot \sep(G) \\
    & \leq (2\ln n +1) \cdot \lceil\log n \rceil \cdot \maxsepRB(G) 
  \end{align*}
  
  Hence, algorithm $A$ is a polynomial-time $O(\ln^2 n)$-approximation algorithm for \PBMAXSEP.
\end{proof}

\section{Conclusion}\label{sec:conclu}

We have initiated the study of \PBSEP{} and \PBMAXSEP{} on graphs, problems which seem natural given the interest that their geometric version has gathered, and the popularity of its ''non-colored'' variants \textsc{Identifying Code} on graphs or \textsc{Test Cover} on set systems.

When the coloring is part of the input, the solution size of \PBSEP{} can be as small as 2, even for large instances; however, we have seen that this is not possible for \PBMAXSEP{} since $\maxsepRB(G)\geq\lfloor\log_2(n)\rfloor$ for twin-free graphs of order $n$. Moreover, $\maxsepRB(G)$ can be as large as $n-1$ in general graphs, yet, on trees, it is at most $2n/3$ (we do not know if this is tight, or if the upper bound of $3n/5$, which would be best possible, holds). However, the upper bound of Theorem \ref{thm:TreeBound}, which is based on the number of support vertices, is tight for trees. It would also be interesting to see if other interesting upper or lower bounds can be shown for other graph classes.

We have shown that $\sep(G)\leq\lceil\log_2(n)\rceil\cdot\maxsepRB(G)$. Is it true that $\sep(G)\leq 2\maxsepRB(G)$? As we have seen, this would be tight.

We have also shown that \PBMAXSEP{} is NP-hard, yet it does not naturally belong to NP. Is the problem actually hard for the second level of the polynomial hierarchy?

\end{document}